\documentclass[letterpaper,11pt]{amsart}
\usepackage{latexsym,array,delarray,amsthm,amssymb,epsfig}%eufrak
\usepackage{caption}
\usepackage{subcaption}

\theoremstyle{plain}
\newtheorem{thm}{Theorem}[section]
\newtheorem{lemma}[thm]{Lemma}

\newtheorem{cor}[thm]{Corollary}

\newtheorem*{thm*}{Theorem}
\newtheorem*{lemma*}{Lemma}
\newtheorem*{prop*}{Proposition}
\newtheorem*{cor*}{Corollary}
\newtheorem*{conj*}{Conjecture}

\theoremstyle{definition}
\newtheorem{defn}[thm]{Definition}
\newtheorem{ex}[thm]{Example}

\theoremstyle{remark}

\newcommand{\zz}{\mathbb{Z}}
\newcommand{\nn}{\mathbb{N}}

\newcommand{\rr}{\mathbb{R}}
\newcommand{\cc}{\mathbb{C}}

\newcommand{\ind}{\mbox{$\perp \kern-5.5pt \perp$}}

\begin{document}

\title{Identifiability of 3-Class Jukes-Cantor Mixtures}
\author{Colby Long and Seth Sullivant}

             \email{celong2@ncsu.edu}
              \email{smsulli2@ncsu.edu } 
               
              \address{Department of Mathematics, Box 8205, North Carolina State University, Raleigh, NC, 27695-8205, USA }  

\maketitle
\begin{abstract}
We prove identifiability of the tree parameters of
the $3$-class Jukes-Cantor mixture model.  The proof
uses ideas from algebraic statistics, in particular: finding
phylogenetic invariants that separate the varieties
associated to different triples of trees; computing
dimensions of the resulting phylogenetic varieties;
and
using the disentangling number to reduce to trees with a small
number of leaves.  Symbolic computation also plays a
key role in handling the many different cases and finding
relevant phylogenetic invariants.
\end{abstract}

\section{Introduction}
\label{Introduction}

% 7) Edit Dimension Section
 % QUESTIONS: Dimension section
 %change multiple citations for alphabetical/numerical 

A phylogenetic model is a statistical model of the evolutionary relationships among a group of taxa. A standard feature of these models is a tree parameter which is meant to encode the common ancestry of the taxa under consideration. If a model produces a probability distribution consistent with observed data, one would like to be able to infer the true phylogeny from the probability distribution. In order for this to be possible, a distribution arising from the model must uniquely determine the tree parameter that produced it. In other words, one would like to be able to determine if the tree parameter of the model is \emph{identifiable}.

The identifiability of both the tree parameter and the continuous parameters has already been established for the basic models of character evolution \cite{Chang} as well as some of the more complex models \cite{Allman2008, Allman2008a,Allman2006}. The more complex phylogenetic models incorporate specific information about the mechanisms of evolution. For example, mixture models are designed to account for biological phenomena that result in data from more than one tree. 
A number of papers have examined the identifiability of mixture models with various restrictions on the topology of the tree parameters \cite{Allman2012, MM, MM2007, Rhodes}. Recent work on two-class mixture models has established the identifiability of the tree parameters for 2-class mixture of both the Jukes-Cantor and Kimuara 2-parameter  models \cite{Allman}. 
Our goal in the present paper is to extend the ideas from \cite{Allman} 
to larger class mixture
models, in particular, we extend the results for $3$-tree Jukes-Cantor mixtures.
Our main result is the following:

\begin{thm}\label{MainThm}
The tree parameters of the $3$-class Jukes-Cantor mixture model are
identifiable on trees with $\geq 6$ leaves.
\end{thm}

The proof of this main result will occupy the whole of the present paper and
uses tools from algebraic geometry and combinatorics as well as some heavy symbolic computation.

In Section \ref{Preliminaries}, we will demonstrate why algebraic geometry is the appropriate tool for studying these models by associating to each set of tree parameters an algebraic variety of the possible distributions arising from those trees. We will then show how the question of identifiability can be reduced to showing that for any two sets of tree parameters, the associated varieties are not contained in one another. To show the varieties are not contained in one another, it is  enough to show that their vanishing ideals are not contained in one another. Elements of the vanishing ideal are called \emph{phylogenetic invariants}, and isolating these invariants will be a key part of this proof.  

In Section \ref{Disentangling} we will investigate the combinatorial properties of binary leaf-labelled trees to show that it is not actually necessary to compare \emph{arbitrary} sets of tree parameters. Instead, we will be able to obtain general identifiability results for $n$-leaf trees by comparing mixtures on trees with $\leq 6$ leaves.  Thus, we will have a finite list of pairs of mixtures for which we must show the mutual noncontainment of their varieties.

In Section \ref{Fourier} we will introduce the Fourier-Hadamard transformation which will simplify the parameterization of our varieties by turning them into toric varieties. We will use these new coordinates in Section \ref{Dimension} to  show that for $n$-leaf trees, all of the varieties under consideration have the same dimension.  This observation halves the number of phylogenetic invariants we must find to separate two mixtures.

Finally, in Section \ref{Phylogenetic Invariants}, we will combine the results from the previous sections to construct a finite list of specific pairs of mixtures that we must consider.  We will then outline a method for finding phylogenetic invariants that distinguish these mixtures from one another and provide access to computations proving that they exist.  Many, but not all
 pairs of triplets of trees are separated by linear invariants.  For the triplets not separated by linear invariants, we will use the linear invariants in a novel way to construct separating invariants of higher degree. 
   
  %----------------------------------------------------------------------------------- 
%
%    PRELIMINARIES
%%----------------------------------------------------------------------------------- 

\section{Preliminaries}
\label{Preliminaries}

The tree parameter of a phylogenetic model of $k$-state character change on an $n$-taxon tree is a binary leaf-labelled tree $T$ with label set $[n]$. To each node $v$ of the tree, we associate a random variable $X_v$ that can take on any of the $k$ characters. For the unique edge $e$ between $v$ and $w$, the matrix $A^e$ is the  matrix with $(i,j)$-th entry equal to $P(X_v= i | X_w = j)$. In other words, this matrix encodes the probability of a character transition along the edge $e$.  Note that there are $k^n$ possible character states at the leaves, and we can calculate the probability of observing a particular state by summing over all possible states of the internal nodes. 
The continuous parameters of such a model are the entries of these matrices,
 and for each choice of parameters we obtain a probability distribution at the 
 leaves. Therefore, we have a map from the stochastic parameters $\Theta_{T}$ into the probability simplex $\Delta^{k^n-1}$, 
$$
\psi_T : \Theta_T \rightarrow \Delta^{k^n-1}  \subseteq  \rr^{k^{n}}.
$$
Since every element of ${\rm im}(\psi_T$) is a probability distribution, all of the entries must be between zero and one and must sum to one. 
Each coordinate function of $\psi_T$ is a polynomial map in the continuous parameters. The degree of each polynomial is equal to the number of edges in the tree parameter and the number of terms is equal to $k$ raised to the number of internal vertices.  Therefore, if we ignore the restrictions on the domain and range and simply regard $\psi_T$ as a complex polynomial map,  
$\overline{{\rm im}(\psi_T)} = V_T$ is an algebraic variety.

An $r$-class mixture model enlarges the space of possible distributions by taking $r$ tree parameters as input and introducing $r-1$  \emph{mixing parameters}.  The mixing parameters weight the distribution from each of the trees according
 to the proportion of data arising from that tree. 
Note that the underlying trees might not be distinct. 
This could be the case where the tree topologies are the same, but the entries of the transition matrices are not the same along correpsonding edges. 
Just as before, with fixed tree parameters, we have a map that takes a choice of continuous parameters for each tree and a choice of mixing parameters and sends them to a probability distribution. As an example, consider the map for a  3-class mixture, which will be our primary object of interest.
$$
\psi_{T_1,T_2,T_3}: \Theta_{T_{1}} \times \Theta_{T_{2}} \times \Theta_{T_{3}} 
\times \Delta^{2} \rightarrow \Delta^{k^n - 1} 
$$
where 
$$
(s_1, s_2, s_3, \pi) \mapsto \pi_1 \psi_1(s_1) + \pi_2 \psi_2(s_2) + \pi_{3}\psi_3(s_3).
$$
Here, $\pi = (\pi_{1}, \pi_{2}, \pi_{3}) \in \Delta^{2}$ is the vector of
 mixing parameters.  Again, regarded as a complex polynomial map, $\overline{{\rm im}(\psi_{T_1,T_2, T_3})}$ is an algebraic variety. 
 In fact, 
$$
\overline{{\rm im}(\psi_{T_1,T_2, T_3})}= V_{T_1}*V_{T_2}*V_{T_3},
$$
where $V_{T_1}*V_{T_2}*V_{T_3}$ denotes the \emph{join} variety of $V_{T_1}, V_{T_2},$ and $V_{T_3}$. 

Before we formally define the concept of identifiability for $r$-class mixtures, we will introduce some convenient notation.  Let $\mathcal{T}_{X}$ be the set of trivalent leaf-labelled trees with label set $X$ and let 
$\mathcal{T}_{X,r}$ be the set of unordered lists of elements of $\mathcal{T}_{X}$ of  length $r$ (i.e.~$r$ element multisets). 
Note that as in our mixture models, for $T= (T_1, \ldots, T_r)$, the trees in $T$ are not necessarily distinct. 
We will now write $ \psi_T := \psi_{T_1, \ldots, T_r }.$

\begin{defn} 
\label{Identifiability}
\cite{Allman}
The tree parameters of an $r$-tree mixture model are \emph{generically identifiable} for $n$ leaf trees if for all $T \in \mathcal{T}_{[n],r}$ and generic choices of $(s_1, \ldots, s_r, \pi) \in \Theta_{T_{1}} \times \cdots
\times  \Theta_{T_{r}} \times \Delta^{r-1}$,
if there is a $T'$ and 
$(s'_1, \ldots, s'_r, \pi') \in \Theta_{T'_{1}} \times \cdots
\times  \Theta_{T'_{r}} \times \Delta^{r-1}$ such that
$$
\psi_{T}(s_1, \ldots, s_r, \pi) = 
\psi_{T'}(s'_1, \ldots, s'_r, \pi')
$$
then $T = T'.$
\end{defn}

In this paper we will not need such generality, as we will specifically consider the 3-class Jukes-Cantor mixture model. The Jukes-Cantor model is a 4-state character substitution model of DNA evolution, with states A,C,G, and T corresponding to the DNA bases. We assume equal transition probabilities between characters, so the transition matrices have the form
$$
\begin{pmatrix}
\alpha  & \beta & \beta & \beta \\
\beta  & \alpha & \beta & \beta \\
\beta  & \beta & \alpha & \beta \\
\beta  & \beta & \beta & \alpha \\
\end{pmatrix}.
$$
Because the entries of each row of the transition matrix must sum to one, we essentially have one parameter along each edge, though we will often ignore this in order to homogenize the parameterization. In this context, we can think of the parameter value on an edge as representing edge length, with greater $\beta$ values corresponding to a higher probability of mutation and longer branches.

In order to prove Theorem \ref{MainThm}, we will translate this statement about identifiability into one about algebraic varieties.
\begin{lemma}
\label{noncontainment}
\cite{Allman}
Suppose $T_1, T_2, T_3, S_1, S_2,$ and $S_3$ are binary, $n$-leef trees, not necessarily distinct, then for the 3-tree Jukes-Cantor mixture model, 
$V_{T_1} * V_{T_2} * V_{T_3} \not \subseteq V_{S_1} * V_{S_2} * V_{S_3}$ and 
$V_{S_1} * V_{S_2} * V_{S_3} \not \subseteq V_{T_1} * V_{T_2} * V_{T_3}$ implies that the set of stochastic parameters mapping into $V_{S_1} * V_{S_2} * V_{S_3}  \cap V_{T_1} * V_{T_2} * V_{T_3}$ is a set of Lebesgue measure 0.
\end{lemma}

This algebraic characterization means that we are able to obtain results about the stochastic parameters by working with complex varieties and all the tools thereof. One strategy for proving generic identifiability of the 3-class Jukes-Cantor mixture model on $n$-taxa is then clear. We can simply list all elements of $\mathcal{T}_{[n],3}$ (which we will call
 \emph{n-leaf triplets}) and generate the corresponding varieties.
By Lemma \ref{noncontainment}, if we can show that any two of these varieties are mutuallly noncontained, then we will have established identifiability in the $n$-leaf case. As alluded to in the introduction, we will actually want to look at elements of $I(V_{T_1} * V_{T_2} * V_{T_3})$, which we call the
 phylogenetic invariants of $V_{T_1} * V_{T_2} * V_{T_3}$, (or occasionally just phylogenetic invariants of $T= \{T_1, T_2, T_3\}$, or of the 
mixture model). 
Therefore, for each $(S,T) \in \mathcal{T}_{[n],r}\times \mathcal{T}_{[n],r}$ with $S \not = T$, we need to find an invariant of $T$ that is not an invariant of 
$S$, and vice versa. Once we have done this for a specific pair, we will say that we have \emph{separated} $S$ and $T$. 

This gives us a clear procedure for determining identifiability, but with some obvious drawbacks. First, the 
number of binary phylogenetic trees on $n$-taxa is $(2n - 5)!! = 1\cdot 3 \cdot 5 \cdots  (2n-5)$, which makes generating all possible 
3-tree mixtures computationally prohibitive even for relatively small $n$. 
Secondly, on the face of it, this brute force approach does not seem to offer any way of establishing generic identifiability of the model for arbitrary
 $n$. However, as we will see in the next section, it is possible to establish generic identifiability of the 3-tree Jukes-Cantor mixture model for all $n$ by separating only a finite number of mixtures.

%----------------------------------------------------------------------------------- 
%
%    DISENTANGLING 3-TREE MIXTURES
%%----------------------------------------------------------------------------------- 

\section{Disentangling 3-Tree Mixtures}
\label{Disentangling}

In this section we explain how to use trees with few leaves to establish identifiability.
The size of the trees we need to consider is bounded by the 
disentangling number for phylogenetic mixtures.
For $T \in \mathcal{T}_{X}$ and $K \subset X$, let $T_{| K}$ be the tree obtained by supressing all degree two vertices in the subtree of $T$ induced by the leaves labelled by $K$. For  $T = (T_1, \ldots, T_r) \in \mathcal{T}_{X,r} $, $T_{ |K} = (T_{1|K}, \ldots, T_{r | K} ). $

\begin{ex}
Consider
$T \in \mathcal{T}_{[8]}$ pictured below and $K = \{ 2,3,5,7,8 \}$.

\begin{minipage}{.5\textwidth}
  \centering
\includegraphics[width=.47\linewidth]{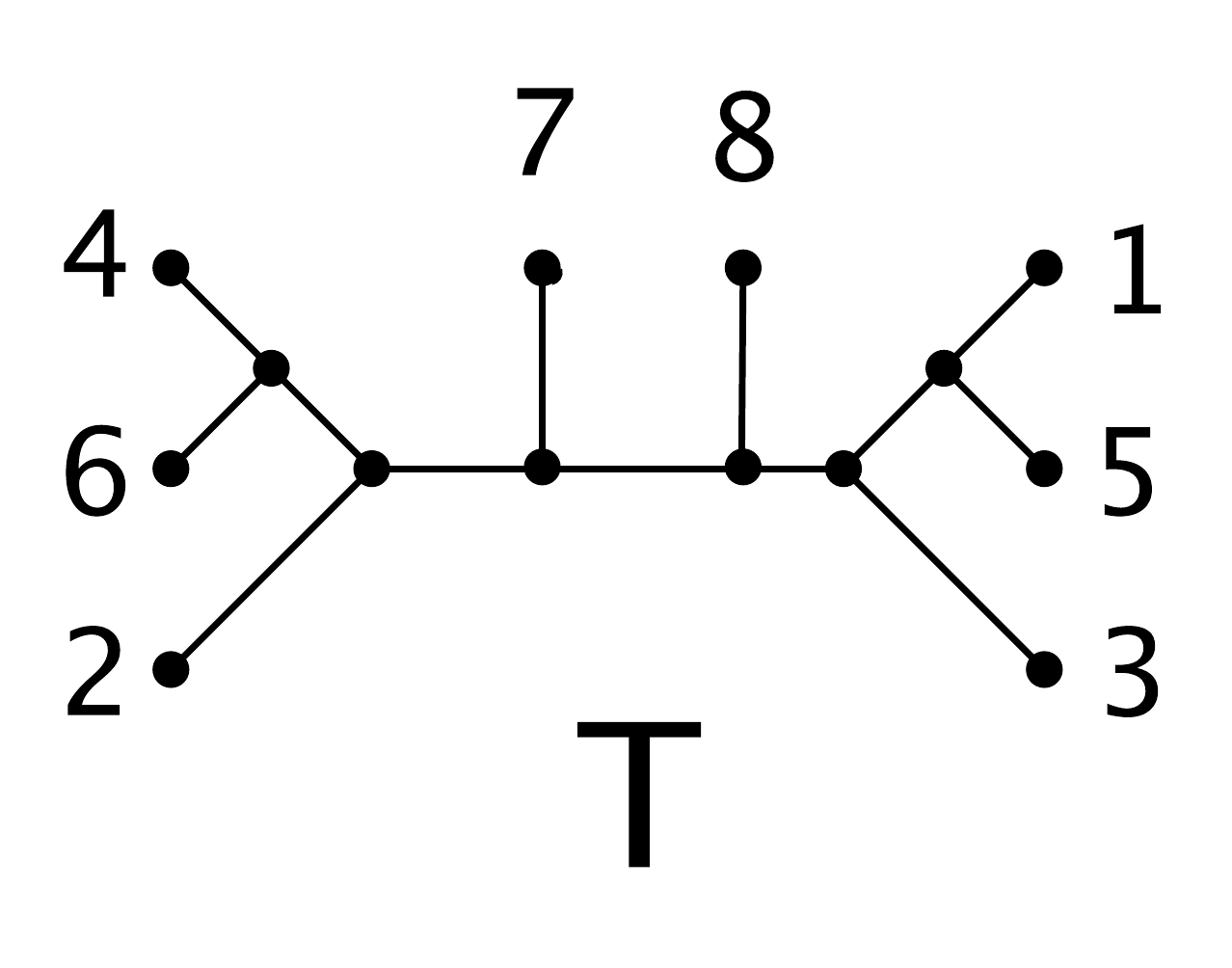}
\end{minipage}
\begin{minipage}{.5\textwidth}
  \centering
  \includegraphics[width=.45\linewidth]{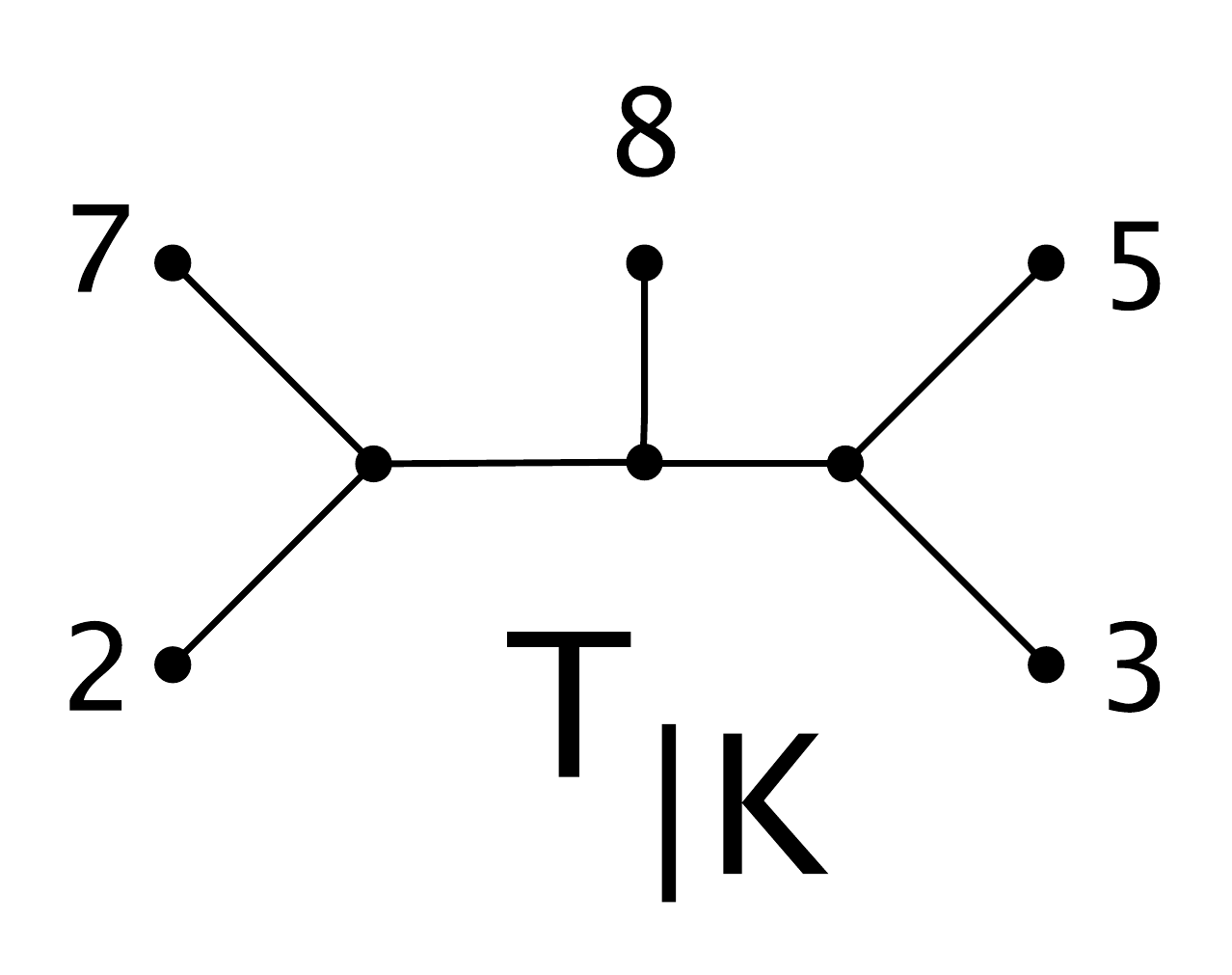}
\end{minipage}
\end{ex}

\begin{defn}
% \cite{DisentanglingNumber}
Let $ S, T \in \mathcal{T}_{X , r} $ with $ S \not = T $. A subset 
$ K \subseteq X $ is said to 
disentangle $ S $ and $ T $ if $ S_{ | K} \not = T_{  | K}.$  Let 
$ d ( S ,  T) $ be the cardinality of the
 minimum disentangling set of $S$ and $T$. The disentangling number $ D (r) $ is 
$$
D(r) = \displaystyle \max_{n \in \nn} \displaystyle \max_{ S \not = T \in \mathcal{T}_{ [n] , r } } d( S, T)
$$
\end{defn}
 
\noindent

The following lemma \cite{Allman} motivates our interest in the disentangling number. 

\begin{lemma}
\label{marginalization}
Let $S,T \in \mathcal{T}_{[n],3}$ and $K \subseteq [n]$. If 
$V_{T_{1|K}}*V_{T_{2|K}}*V_{T_{3|K}} \not \subseteq V_{S_{1|K}}*V_{S_{2|K}}*V_{S_{3|K}}$ then $V_{T_1 }*V_{T_2  }*V_{T_3 } \not \subseteq V_{S_1  }*V_{S_2  }*V_{S_3 }.$
\end{lemma}

\noindent
Now suppose we are able to show identifiability for trees with $D(3)$ leaves. Then given any two mixtures $S,T$ on $n > D(3)$ leaves we can find some $K \subset [n]$ such that $T_{|K} \not = S_{|K}$ and $V_{T_{1|K}}*V_{T_{2|K}}*V_{T_{3|K}} \not \subseteq V_{S_{1|K}}*V_{S_{2|K}}*V_{S_{3|K}}$. By Lemma \ref{marginalization}, in doing so we will have 
separated $S$ and $T$. Consequently, we would have generic identifiability of the tree parameter of the 3-class Jukes-Cantor mixture model for all $n \geq D(3)$. Thus, as promised, we will have an upper bound on the number of possible varieties we need to consider. In this section we provide some general background on the disentangling number and prove that $D(3)=6$.

The \emph{rooted} disentangling number, $ RD(r) $, is defined analogously for rooted trees. We will omit the short proof of this lemma  from \cite{DisentanglingNumber} that 
relates $RD(r)$ and  $D(r)$.

\begin{lemma} 
The disentangling and rooted disentangling numbers satisfy: $ D(r) \leq RD(r) + 1$.
 \end{lemma}
 
The main result of \cite {DisentanglingNumber} is the following theorem from which we obtain an upper bound on $D(r)$ as an immediate corollary. 

\begin{thm}
\label{Rooted}
$ RD(r) =  3( \lfloor \text{log}_2(r) \rfloor + 1 )$ .
 \end{thm}
 
 \begin{cor}
For $ r \in \nn ,  D(r) \leq  3( \lfloor \text{log}_2(r) \rfloor + 1 )  +1 $.
\end{cor}

The original proof Theorem \ref{Rooted} is obtained by encoding multisets of
trees as high-dimensional contingency tables and applying 
results about marginal maps. We provide an alternative, 
and hopefully more direct proof by examining the tree topologies directly. 
 
\begin{proof}[Proof of Theorem \ref{Rooted}  ] A construction in \cite{Humphries} shows that $  3( \lfloor \text{log}_2(r) \rfloor + 1 ) \leq RD(r) $, so we need only show
that for every pair $ S, T \in \mathcal{T}_{X , r } $ there is a disentangling set of cardinality less than or equal
to $  3( \lfloor \text{log}_2(r) \rfloor + 1 ) $. We will proceed by induction on $ r $. Because a rooted tree is determined by its rooted triples
 (\cite[Theorem 6.4.1]{SempleandSteel}), the base case $ RD(1) = 3 $ is established. Assume this is true
 for all integers less than $ r $ and let $ S = ( S_1, \ldots, S_r )$ and $ T = ( T_1, \ldots, T_r ) $ 
 be two unordered lists of rooted binary trees with $ S \not = T $. There
 must exists some $ S_i $ and $ T_j $ such that $ S_i \not = T_j $. By our inductive assumption, 
 we can permute the leaf labels so that for $ K = \{1, 2, 3 \} $, $ T_{ i | K } \not = S_{ j | K} $.
 There are only three topologically distinct rooted leaf-labelled binary trees with label set $K$, which we will
 label $ t_1, t_2, $ and $ t_3 $. 
 
If $ S_{ | K } \not = T_{ | K } $, then $ S $ and $ T $ are disentangled and we are done. 
Otherwise, $ S_{ | K } = T_{ | K } $ is an unordered list of the trees $ t_1, t_2, $ and $ t_3 $
occurring with multiplicity. Partition $ S $ into three multisets,
$$
L_S^l : = \{ S_j \in S |  S_{j | K} = t_l \} 
$$

\noindent
for $ 1 \leq l  \leq 3 $,
and likewise for $T$. Since $K$ was chosen to disentangle
an element of $S$ from an element of $T$, it must be the case
 that $ S_{ | K } $ and $ T_{ | K } $
contain at least two distinct three leaf trees. Therefore, we can 
choose $l$ so that
 $ L_S^{l} $ is  nonempty and $ | L_S^{l} | = r' \leq  \frac{r}{2} $. 
 Since $ S_{ | K } \not = T_{ | K } $,
$ | L_S^l | =  | L_T^l | $ and we can consider  $L_T^l$  and $L_S^l$ as elements of 
 $\mathcal{T}_{X,r'} $. By our inductive assumption,
 there exists a disentangling set $K'$ such that  
 \begin{align*}
 |K'| & \leq  3( \lfloor \text{log}_2(r') \rfloor + 1 ) \\
 & \leq 3( \lfloor \text{log}_2(\frac{r}{2}) \rfloor + 1 ) \\
 & = 3(( \lfloor \text{ log}_2(r) \rfloor - 1 ) + 1 ) \\
 & = 3( \lfloor \text{log}_2(r) \rfloor).
 \end{align*}
Therefore, $ |K \cup K'| \leq  3( \lfloor \text{log}_2(r) \rfloor + 1 ) $. We claim that this set disentangles $S$ and $T$.  Since $K'$ disentangles $L^l_T$ from $L^l_S$, and
 $K' \subseteq K \cup K'$,
  $(L^l_S)_{| K \cup K' } \not = (L^l_T)_{| K \cup K' } $. If $S$ and $T$ are still entangled, then there must be  some tree in $(L^l_S)_{| K \cup K'}$ equal to some tree in  $(L^m_T)_{| K \cup K'}$ with $l \not = m$. But since $K \subseteq K \cup K'$, this is impossible, so $K \cup K'$ disentangles $S$ and $T$.
\end {proof}

While this assures us that $D(3) \leq 7$, we can actually reduce this bound slightly, vastly
reducing the number of 3-tree mixtures we need to consider. For the theorem and proof 
that follow, we will make use of the following definition. 

\begin{defn}
For  $T \in \mathcal{T}_{[n]} $ and $K$ a three element subset of $ [n] $, if $T$ has 
a split that separates $K$ from  $[n] \setminus K $ then
$T_{|K}$ is a \emph{cluster on K}.
\end{defn}

\begin{defn}
Let $ S, T \in \mathcal{T}_{X , r} $ and $K $  a subset of $ X $ that does not disentangle $S$ and $T$.
Label the trees of $S$ and $T$ so that $S = (S_1, \ldots, S_r)$ and $T = (T_1, \ldots, T_r)$.
For $1 \leq m \leq r$, let
$$
 m_i = \text{min} (\{ m \in [r]\setminus \{ m_1, \ldots, m_{i-1} \} : S_{m | K } = T_{ i  | K}\} )
$$
Then with respect to the chosen labeling, we say that $S_{m_i} $ and $T_i$ are \emph{partners} at $K$.
\end{defn}

Notice that each tree of $T$ has exactly one partner at $K$, and that the partnered trees at $K$
are exactly the same if we swap the roles of $S$ and $T$ in the definition.

\begin{figure}
\centering
\begin{minipage}{.49\textwidth}
  \centering
  \includegraphics[width=.3\linewidth]{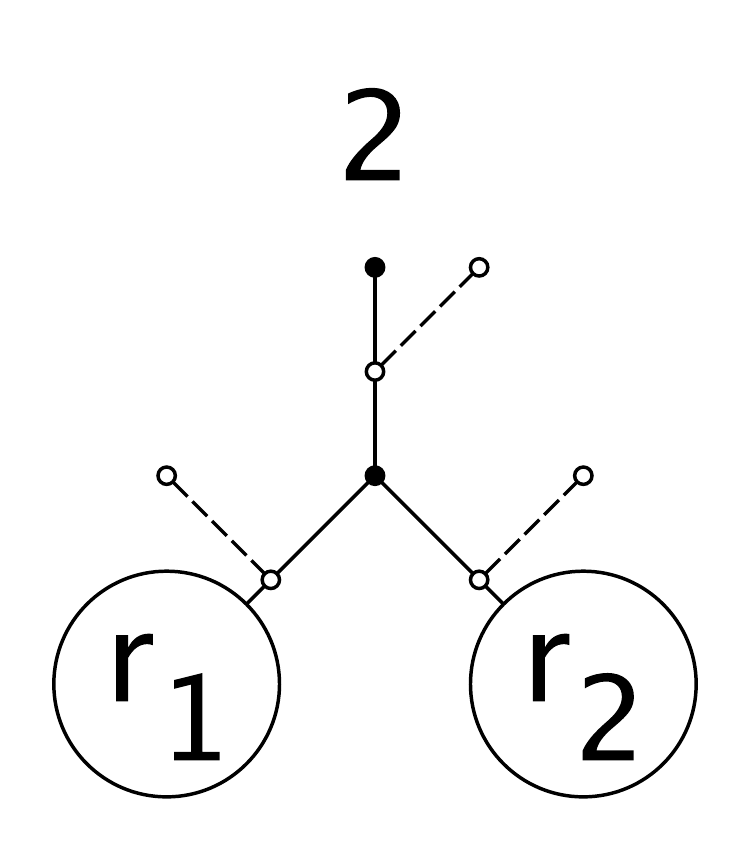}
  \captionof{figure}{Possible locations for $e_1$.}
  \label{NewLeaf}
\end{minipage}
\begin{minipage}{.49\textwidth}
  \centering
  \includegraphics[width=.5\linewidth]{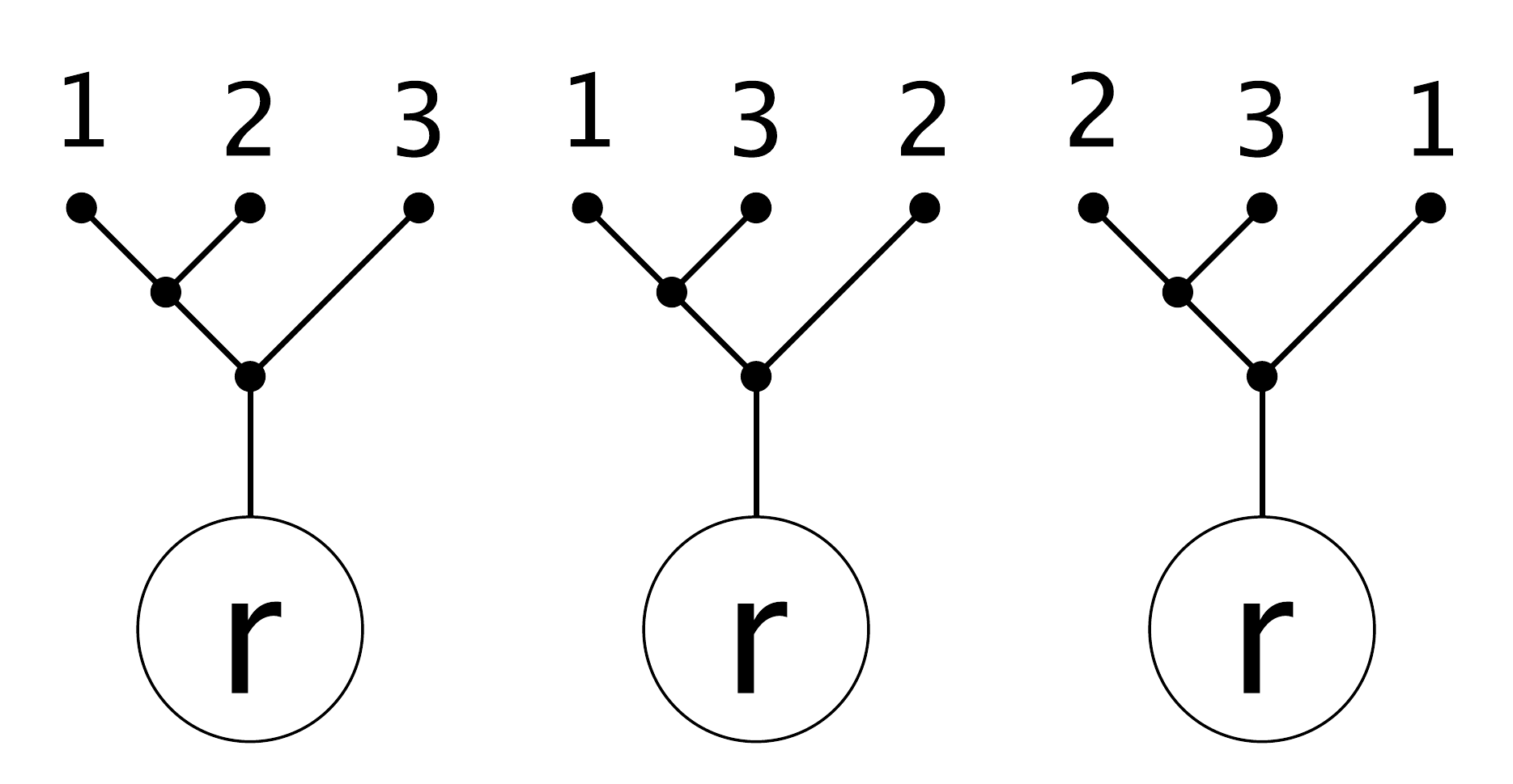}
  \captionof{figure}{Possible structures for 3-partners at $\{K_1, K_2, K_3\}$.}
  \label{Clusters}
\end{minipage}
\end{figure}

\begin{thm}
\label{(D(3)}
$ D(3) \leq 6. $
\end{thm}
 
\begin{proof}
 We will use contradiction. Suppose $ D(r) = 7$ and  let $ K_i = [7] \setminus i . $
 Then there must exist $ S , T \in \mathcal{T}_{ [7], 3 } $ such 
 that $ S_{ | K_i } = T_{ | K_i }$ for $1 \leq i \leq 7$. For everything that follows, fix some labelling of the 
 trees of $S$ and $T$ so that for each $i$, 
 every tree of $T$ and $S$ has a partner at $K_i$.

We will collect a few key observations about trees that 
are partnered together at multiple $K_i$.
If a tree of $S$ and a tree of $T$ are partnered together at $K_i$ for exactly $j$ distinct
values of $i$, then we will call them $j$-partners.
Suppose $ S_l  \not = T_m $ are $2$-partners and permute
the leaf labels so that they are partnered at $K_{1}$ and at $K_{2}$. Let $v_i$ be the leaf vertex labelled 
$i$, and let  $e_i$ be the edge adjacent to this vertex. Since $ S_{ l | K_{1} } = T_{ m  | K_{1} } $, 
there must be a unique edge on this tree where $e_1$ is attached to form $S_l$,
 and a different unique edge where $e_1$ is attached to form $T_m$. 
 But in order for our trees to still be equal when restricted to $K_{2}$, the two distinct
 edges where we attached $e_1$ must collapse to
the same edge when we remove $e_2$. Therefore, $S_l$ and $T_m$ must have the structure of the tree in Figure \ref{NewLeaf} where $r_1$ and $r_2$ are rooted trees and where $e_1$ is one of the dashed edges.

Now suppose that $S_l$ and $T_m$ are 3-partners, partnered at $K_1, K_2$ and $K_3$. 
From above, observe that
regardless of which edge is $e_1$, that the length of the path from $v_1$ to $v_2$ in both trees 
must be less than or equal to three. Therefore, the length of the path between each pair of vertices, 
$(v_1, v_3)$, $ (v_1 , v_2) $, and $(v_2 ,v_3)$, must be less than three in both $S_l$ and $T_m$.
 Consequently,  $S_l$ and $T_m$ must both have a cluster on $ \{ 1, 2, 3 \} $, 
 and they must be the same tree apart from these
clusters (i.e., $S_l$ and $ T_m $ must be two different trees from the list in Figure \ref{Clusters}, where $r$ is some rooted tree).
From the figure we also see if  $ S_{ l | K_{i} }  =  T_{ m | K_{i} } $ for $i \not \in \{ 1, 2 , 3 \}$, then $S_l = T_m$. This implies that for $j >3$ any two trees that are $j$-partners must be the same tree.

 If any tree of $S$ is equal to any tree of $T$, then
we can remove these trees to form the lists $S', T' \in \mathcal{T}_{ [7], 2 } $, and any set $K$ 
that disentangles $S'$ and $T'$ will disentangle $S$ and $T$. Since $D(2) = 6$ (\cite{MM})
this would imply $d ( S , T )  \leq 6$ contradicting our assumption that $d(S,T) = 7$. Therefore, we can assume that 
no two trees are $j$-partners for $j  > 3$. Since each tree of $S$ and each tree of $T$ must be
 partnered at all seven $K_i$, the only possibility for a single tree is that it has one 3-partner
 and two 2-partners or two 3-partners and one 1-partner. The particular partnering relationships impose restrictions on the possible structures of the trees in $S$ and $T$. We will now consider both cases and  
use these restrictions to arrive at a contradiction.

\bigskip

\noindent
\emph{Case 1:}  There exists a tree in $S$ or $T$ with one 3-partner and two 2-partners.

\begin{figure} 
\centering
\begin{minipage}{.49\textwidth}
  \centering
  \includegraphics[width=.45\linewidth]{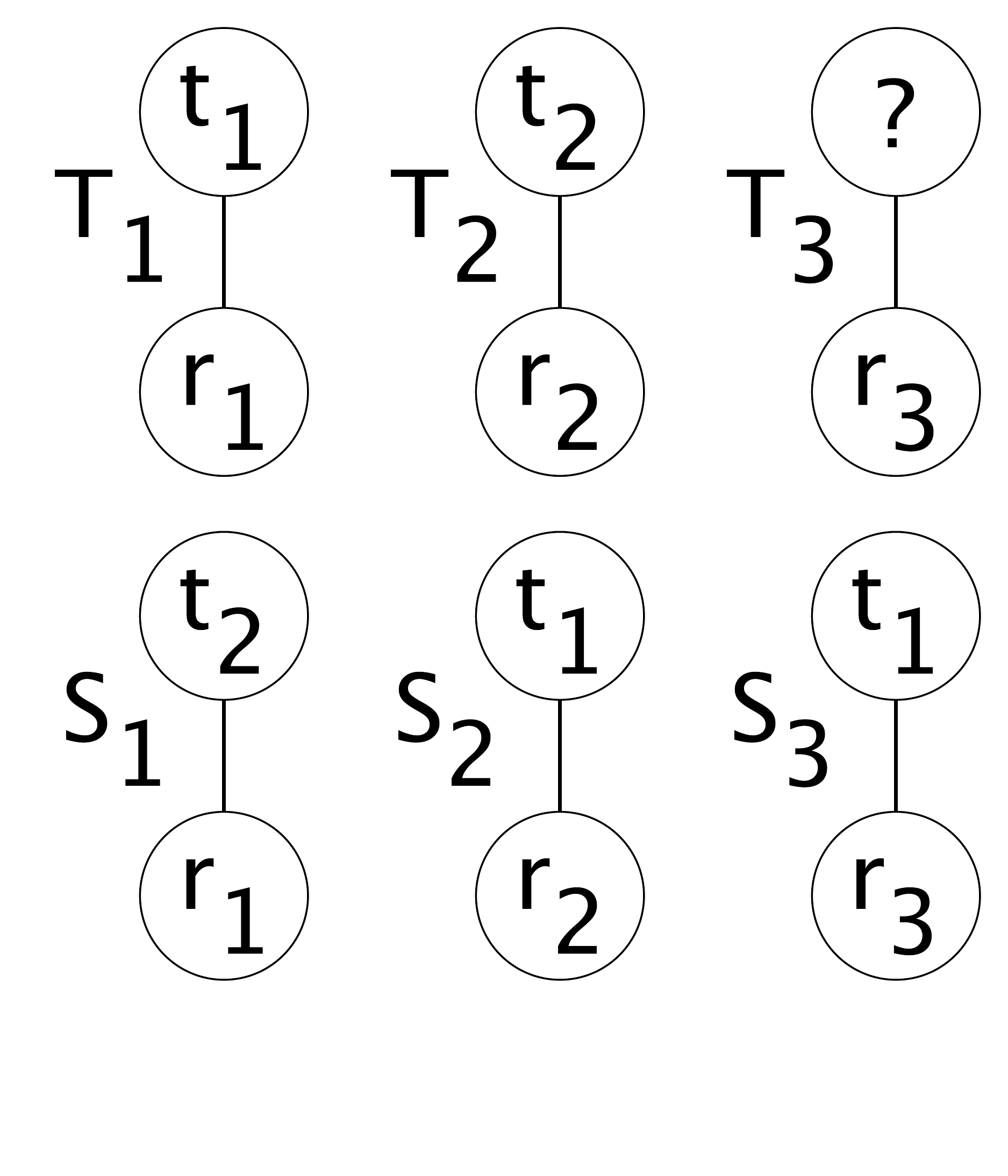}
   \captionof{figure}{Structure of trees satisfying Case 1.}
   \label{One3Partner_B}
\end{minipage}
\begin{minipage}{.49\textwidth}
  \centering
  \includegraphics[width=.45
   \linewidth]{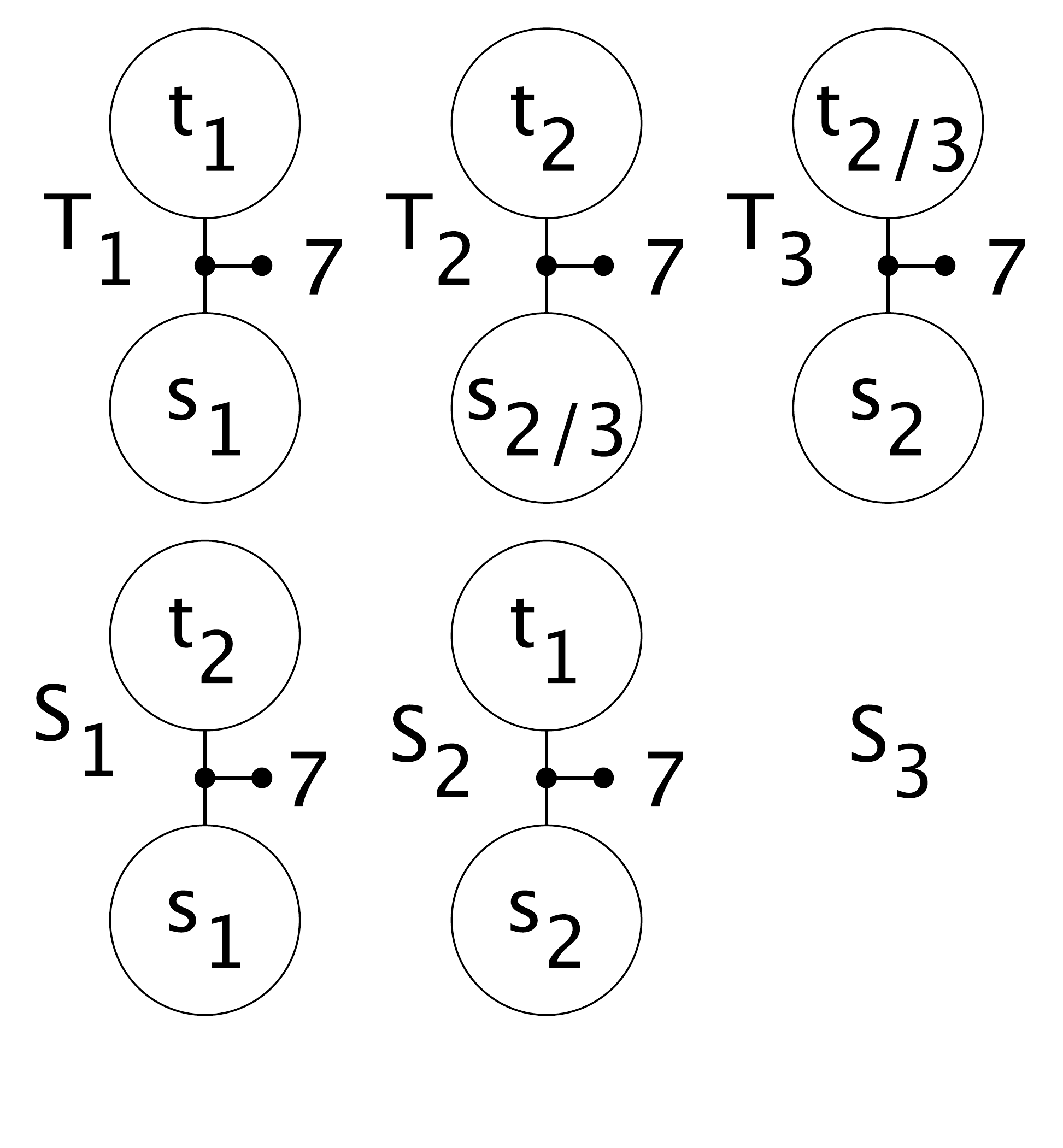}
   \captionof{figure}{Structure of trees satisfying Case 2.}
   \label{Two3Partners} 
 \end{minipage}
\end{figure}

\bigskip

\noindent
We will leave all partners fixed according to the original labelling. However, for convenience, we will relabel the lists, 
trees, and leaves so that
$T_1$ is partnered with $S_1$ at $ \{ K_1, K_2, K_3 \} $, with $ S_2 $ at $ \{ K_4, K_5 \} $, and with $ S_3 $ at $\{ K_6, K_7 \}$. 

So far then, we know $T_1$ and $S_1$ are 
as in Figure \ref{One3Partner_B}, where $t_1$ 
and $t_2$ are distinct clusters on $\{ 1, 2, 3 \} $
and $r_1$ is a rooted tree with leaf label set $ \{ 4, 5, 6, 7 \} $.
We also know that $T_{1 |K_4} = S_{2 |K_4} $, 
so $S_{2 |K_4} $
contains the cluster $t_1$. If $e_4$ is connected to an edge of $S_2$ somewhere in the cluster $t_1$ of $S_{2 |K_4} $, then it is impossible for $S_{2 | K_5} = T_1$. Therefore, we see that even without restricting to $K_4$,
$S_2$ and similarly $S_3$ must contain the cluster $t_1$. Thus, 
$T_1, S_1, S_2,$ and $S_3$ are all as depicted in Figure \ref{One3Partner_B}. 

Every tree must have a 3-partner, so let $T_2$ be a 3-partner of $S_2$. 
From our observations above, $S_2$ and $T_2$ differ only by a cluster on some three element set $K'$.
We know that $K'$ cannot contain $4$ or $5$ since $S_2$ is partnered with $T_1$ at $K_4$ and $K_5$.
Therefore, $K'$ contains at least one element of $ \{ 1, 2, 3 \}$. 
But to preserve the cluster $t_1$, it must be that $ K' = \{1, 2, 3 \} $. If $T_{2|K'} = t_3$, then $K'$ disentangles
$S$ and $T$, and if $T_{ 2 | K' } = t_1$ then $T_2 = S_2 $, it follows that  $T_{2 | K' } = t_2$.

Finally, $S_1$ and $T_1$ as well as $S_2$ and $T_2$ are 
partnered at $K_1, K_2,$ and $K_3$, which forces $S_3$ and $T_3$ 
to be partnered at $K_1, K_2$, and $K_3$.  As a result, $S_3$ and $T_3$ differ only by a cluster on $ \{1, 2 , 3 \} $ (Figure \ref{One3Partner_B}).
If $T_{ 3 | K' } = t_1$ then $T_3 = S_3$ and if $T_3 = t_2$ or $T_3 =t_3$ then $K'$ is a disentangling set. In any case we have a contradiction.

\bigskip

\noindent
\emph{Case 2:}  Every tree in $S$ and $T$  has two 3-partners and one 1-partner.

\bigskip

\noindent 
As before, leave the partnering relationships fixed and relabel the lists, trees, and leaves so that
$T_1$ is partnered with $S_1$ at $ \{ K_1, K_2, K_3 \} $, with $ S_2 $ at $ \{ K_4, K_5, K_6 \} $, and with $ S_3 $ at $\{ K_7 \}$. 
As we've seen, $T_1$ must have a cluster on $\{ 1, 2, 3\}$ and a cluster on $\{ 4, 5, 6 \}$, so $T_1$ is as pictured in Figure \ref{Two3Partners}. As 3-partners of $T_1$, both $S_1$ and $S_2$ must have clusters on  $\{ 1, 2, 3\}$ and $\{ 4, 5, 6 \}$ as well and are also as depicted in Figure \ref{Two3Partners}.

$S_1$ must have another 3-partner which we will label $T_2$. $S_1$ and $T_2$ must
differ only at a cluster on some three element set $K' \subset [7] \setminus \{1, 2 , 3\}$ which must contain elements of $\{4, 5 , 6\}$. Then as we argued above, to preserve the $\{4,5,6\} $
cluster on $S_1$ it must be that $K' = \{ 4, 5, 6 \} $. $S_2$ must also have a second 3-partner. This 
tree can't be $T_2$, since then both $T_1$ and $T_2$ would have two 3-partners, leaving 
$T_3$ as the sole 3-partner for $S_3$ and putting us back in Case 1. Therefore, $S_2$ and $T_3$ must be 3-partners, and the same logic 
shows that they differ only at a cluster on $\{ 1, 2, 3 \}$. The possible structures of $T_1, T_2, T_3, S_1, $ and $S_2$ are all displayed in Figure \ref{Two3Partners} ($t_{i/j}$ indicates that a cluster can be only either $t_i$ or $t_j$ and likewise for $s_{i/j}$).
Since the 3-partners of $S_1$ are $T_1$ and $T_2$, the 1-partner of $S_1$ must be $T_3$, and $S_1$ and $T_3$ must be partnered at $K_7$. From the diagram, it is clear $S_{1 | K_7}  \not = T_{ 3 | K_7}$,
which is a contradiction.
\end{proof}

%----------------------------------------------------------------------------------- 
%
%    THE FOURIER-HADAMARD TRANSFORMATION
%%----------------------------------------------------------------------------------- 

 \section{The Fourier-Hadamard Coordinate Transformation}
 \label{Fourier}
 
For group-based models the Fourier-Hadamard coordinate transformation is a linear change of coordinates that makes each coordinate function of the parameterization a monomial. Importantly for our purposes, the linearity of the transformation means that it commutes with taking mixtures. In this description, both determining the dimension and finding phylogenetic invariants of the mixture varieties become much simpler. We will present a practical outline demonstrating how to recover the monomials; a thorough explanation of the transform can be found in \cite{Evans}, \cite{Sturmfels}, and \cite{Szekely} .  

Let $B$ and $B'$ be the partition of $[n]$ induced by removing an edge $e$ from a leaf-labelled tree $T$. The set of such splits of $T$ is denoted $\Sigma(T)$ and uniquely determines the tree topology. As a result, we will index the edges of $T$ by the splits that they induce. A phylogenetic model is group-based if there exists a group $G$ and functions $f_{B|B'} : G \rightarrow \rr $ associated to each edge of the $n$-leaf tree parameter $T$, such that when the character states are identified with the elements of $G$, the probability of character change along $e$ is dependent only on the difference between the character states at the endpoints of $e$. In other words, the probability of the endpoints of $e$ being in character states identified with group elements $g$ and $h$ is equal to $f_{B|B'}(g-h).$

Let $p_{g_1,\ldots, g_n}$ be the probability of observing the state $(g_1, \ldots ,  g_n)$ at the leaves of $T$ and let $q_{g_1, \ldots, g_n}$ be the image of this coordinate after the Fourier-Hadamard transformation. For the Jukes-Cantor model, we make the following identification of the character states with the elements of $\zz_2 \times \zz_2$: $A \mapsto (0,0), C\mapsto(1,0), G\mapsto(0,1), T\mapsto(1,1)$. Since the JC model assumes all transition probabilities are equal, our concern is only whether or not there are different 
characters at the endpoints of each edge. Thus,  $f_{B|B'}((1,0)) = f_{B|B'}((0,1)) = f_{B|B'}((1,1))$, and after transformation, the new coordinates of the parameter space are $a^{B|B'}_C=a^{B|B'}_G=a^{B|B'}_T \in (0,1]$ and $a^{B|B'}_A=1$. Then

 \begin{displaymath}
   q_{g_1, \ldots, g_n}  = \left\{
     \begin{array}{lr}
       \displaystyle \prod_{{B|B'} \in \Sigma(T)} a^{B|B'}_{\sum_{i \in B} g_i} & : \displaystyle \sum_{i = 1}^{n} g_i = 0 \\
       0 & : \text{otherwise}
     \end{array}
   \right.
\end{displaymath} 

Notice that in the group $\zz_2 \times \zz_2$, if the leaf elements sum to the identity then for every partition $B|B'$, $\displaystyle \sum_{i \in B} g_i = \displaystyle \sum_{i \in B'} g_i.$ Therefore, the monomial above does not depend on our labelling of the sets of the splits. We call the splits of $T$ associated to leaf edges the 
trivial splits and for simplicity we denote the parameters associated to the leaf edge labelled $i$ by $a^i_g.$  

Note that for each nontrivial coordinate and for each edge either $a^{B|B'}_A$ or $a^{B|B'}_C$ appears in the monomial parameterization of the coordinate. We encode the resulting monomials in tree diagrams as follows. Redraw the tree $T$, but make each edge solid if $a^{B|B'}_C$ appears and dotted if $a^{B|B'}_A$ appears. The solid edges of the diagram form a subforest of $T$ and the number of distinct nontrivial Fourier coordinates are in bijection with the subforests of $T$ \cite{SteelandFu}. 

\begin{ex}
\label{qcaggtg}

Let $T$ be the tree with nontrivial splits given by $\{15|2346\}$, $\{135|246\}$, and $\{1235|46\}$. The parameterization of a particular coordinate as well as the subforest induced by this coordinate are shown.

\begin{minipage}{.5\textwidth}
  \centering
  $$
q_{CAGGTG} = a^1_C a^2_A  a^3_C  a^4_C a^5_C  a^6_C a^{15|2346}_C a^{135|246}_A  a^{1235|46}_A $$
\end{minipage}
\begin{minipage}{.5\textwidth}
  \centering
  \includegraphics[width=.47\linewidth]{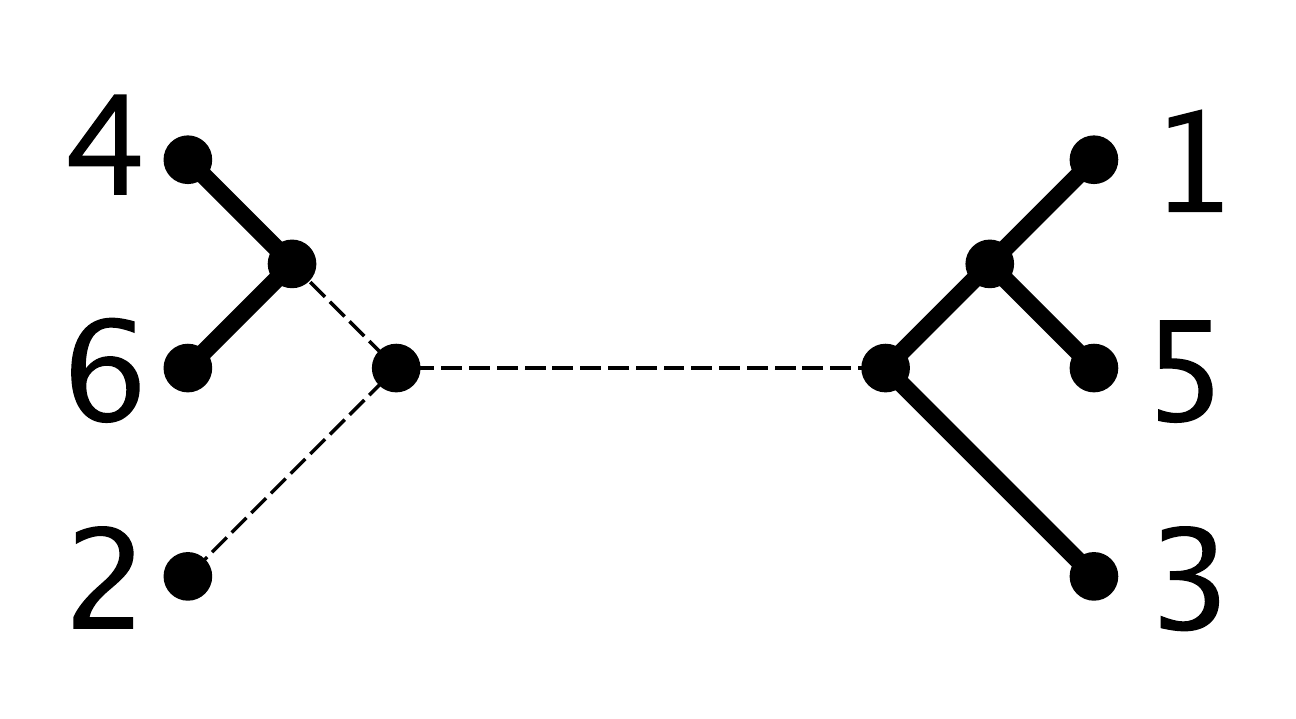}

\end{minipage}
\end{ex}

The inclusion of the $a^{B|B'}_A$ variables homogenizes the parameterization which will be convenient for some the work that follows. However, when writing out the explicit parameterization of a coordinate we will  dehomogenize by setting all of the $a^{B|B'}_A = 1$. This allows us to drop the subscripts entirely. For arbitrary $n$-leaf trees, we will further simplify notation by extending the numbering of the trivial splits to a numbering of all of the splits so that every split is labelled by some element of the set $[2n - 3]$. Thus, in Example \ref{qcaggtg} the coordinate can be written more compactly as $q_{CAGGTG} = a_{1}  a_{3}  a_{4} a_{5}  a_{6} 
a_{7}.$

 %----------------------------------------------------------------------------------- 
%
%    DIMENSION
%%----------------------------------------------------------------------------------- 

\section{Dimension}
 \label{Dimension}
 
Our goal in this section is to prove dimension results for the
appropriate join varieties, which will allow us, for each
pair of triplets, to prove non-containment in only one
direction.  Our technique for deducing the dimension results
relies on the fact the the underlying varieties $V_{T}$ are
toric varieties as explained in the previous section.  We make use
of the tropical secant varieties approach of Draisma \cite{Draisma}
to prove our dimension results and refer the reader to that
paper for background.

When taking the join of $r$ varieties, we introduce 
$ ( r - 1 ) $ parameters, giving us the following bound
 $$
 \text{dim} (V_1 * \ldots * V_r) \leq \text{min} 
 \left \{ \displaystyle\sum_{ i = 1 }^k \text{dim}(V_i) + (r - 1) ,
 \text{dim} (V) \right \},
 $$
where $ V $ is the ambient affine space containing the  join variety. This bound is called the  \emph{expected dimension} and any join variety realizing this bound is called \emph{nondefective}. As an example of how we will use this concept, 
let $T=(T_1,T_2,T_3)$ be a  6-leaf triplet. 
We know that each tree in $T$ has 9 edges, and so from the Fourier parameterization of the Jukes-Cantor model, it is obvious that dim($V_{T_i}) \leq 9$. Therefore, if we can show that dim$(V_{T_1}*V_{T_2 }*V_{ T_3 }) = 29 $, we will have shown that this variety is nondefective. Now suppose we could show that the join variety of every 6-leaf triplet is nondefective, then as a corollary, the join varieties of any two 6-leaf triplets are the same dimension.  Since the varieties involved are irreducible, if they are of the same dimension, the containment
$$
 V_{T_1}*V_{T_2 }*V_{ T_3 } \subseteq  V_{S_1}*V_{S_2 }*V_{ S_3 }
 $$ 
implies that the two join varieties are actually equal. 
Thus, to separate $S$ and $T$, it is enough to show  
$V_{T_1}*V_{T_2 }*V_{T_3 }$ and 
 $ V_{S_1 }*V_{S_2}*V_{S_3} $ are not the same variety, meaning we would only need to find an invariant for one that is not an invariant for the other. 
 Since eventually our proof will necessitate separating certain triplet pairs, establishing nondefectiveness will greatly reduce the number of invariants we have to find. Moreover, it will make the task of separating pairs much easier in cases where finding an invariant to establish noncontainment in one direction is more difficult than doing so in the other.

The following definitions and terminology up to Theorem \ref{nondefective} are adapted from the more general presentation in \cite{Draisma}. Our approach is inspired by proof of the nondefectivess of the second secant variety associated to a 4-leaf tree as shown in (\cite{Allman}). Let $T \in \mathcal{T}_{[n],r}$, we will temporarily regard the $ a^{i}_A $ 
as variables which homogenizes the parameterization of $V_{T_i}$ for $1 \leq i \leq r$. We will consider these projective varieties as affine cones, $C_i$, which are closed under scalar multiplication. Now the additional parameters introduced when 
 constructing the join variety are superfluous, and we can take as our definition that the \emph{join} of $C_1$ through $C_r$  
 is

$$
 C_1 + \ldots + C_r := \{ c_1 + \ldots + c_r | c_i \in C_i, 1 \leq i \leq r\}. 
 $$

Each of these cones is the image of a 
polynomial map, $f_i : \cc^{m_i} \rightarrow  \cc^{q}$, and in our case, the use of Fourier coordinates ensure that each coordinate function is a monomial. Then for $1 \leq i \leq r$, we can write $f_i$ as a list $ (f_{ i , b })^{q}_{b = 1 } $ where each 
 $f_{i , b} = x^{\alpha}$. For example, if we let $r=3$ and $n=6$, then $m_1 = m_2 = m_3 = 18$,  $ x = (a^{1}_C, \ldots, a^{9}_C,  a^{1}_A, \ldots, a^{9}_A ) $, $\alpha \in \{ 0 , 1 \}^{18} $, and $q=4^6=4096$. 

Let  $ l_{ i , b } : \rr^{m_i} \rightarrow \rr $ be the piecewise linear
function defined by 
 $$ v \mapsto \langle v, \alpha \rangle .$$
 For  $ v = (v_1, \ldots, v_r )  \in \displaystyle \prod_{i=1}^r \rr^{m_i} $ and  $ b \in [q] $, we say that $\emph{i wins b at v} $  if

 \begin{enumerate}
 \item  $l_{i,b}(v_i) < l_{j, b}(v_j) $ for all $ j \not = i $ and
 \item $ l_{i,b} $ is differentiable (hence linear) near $v_i$. 
 \end{enumerate}
Denote the corresponding differential by $d_{v_i} l_{i,b}.$ 

\begin{defn}
For $ v = (v_1, \ldots, v_r )  \in \displaystyle \prod_{i=1}^r \rr^{m_i} $, let
$$ 
D_i(v) := \{ d_{v_i} l_{i,b}  | i \text{ wins } b \text{ at } v \} 
$$
be the set of \emph{winning directions of i at v}.
\end{defn}

\noindent Finally, we have all the requisite definitions to state the primary result we will need.

\begin{lemma}
\label{lemma: Draisma}
(\cite{Draisma})
The dimension of $ C_1 + \ldots + C_r $ is at least the maximum, taken over all $ v = (v_1, \ldots, v_r )  \in \displaystyle \prod_{i=1}^r \rr^{m_i} $, of the sum 
$$\displaystyle\sum_{ i = 1 }^r {\rm dim}_{\rr} \langle D_i(v) \rangle_\rr .$$
\end{lemma}

\noindent This gives us a way to compute lower bounds on the dimensions of join varieties. Now to show that a join variety is non-defective, we just need to show that this lower bound is equal to the expected dimension. 

\begin{thm}
\label{nondefective}
Let $T \in \mathcal{T}_{[n],r}$. For $n \geq 4$ and $r \leq  \lceil \frac{n}{2} \rceil$, the join variety $V_{T_1}*\ldots*V_{T_r} $ associated to the $r$-class Jukes-Cantor mixture model is nondefective.
\end{thm}

\begin{proof} 
Let $T \in \mathcal{T}_{[n],r}$, by Lemma \ref{lemma: Draisma}, to show nondefectiveness it will be enough to find a vector $v = (v_1, \ldots , v_r) $ so that
for $1 \leq i \leq r$, ${\rm dim}_{\rr} \langle D_i(v) \rangle_\rr = (2n-3) + 1.$ Thus,  $C_1 + \ldots + C_r $ will have dimension $r(2n-3) + r$, and when we projectivize by setting the $ a^{j}_A = 1 $, we will have  $ {\rm dim} (V_{T_1}*\ldots*V_{ T_r }) = r(2n-3) + (r-1) $ as desired. 

The set of winning directions of $i$ at $v$, $D_i(v)$, is a set of 0-1 vectors in $\rr^{4n-6} $. Our goal will be to construct the vector $v$ in such a way that the vectors in each $D_i(v)$ span a space of dimension $2n - 2$. Recall that for a tree $T$, the distinct Fourier coordinates are in 
bijection with the subforests of $T$. Therefore, each $b$ induces a subforest on the trees $T_1, \ldots, T_r$ the number of leaf edges of which, $t_b$,  depends only on the number of entries that are not $A$ in the index of the $b$-th coordinate. For example, if the $b$-th coordinate is $q_{AACGT}$ then $t_b = 3$ since in \emph{any} 5-leaf tree the subforest induced by $b$ contains 3 leaf edges.

\bigskip
\noindent \emph{Case 1:} $r$ is even
\bigskip 

\noindent Construct the vector $v = (v_1, \ldots, v_r)$ as follows.
Each $v_i$ has $4n-6$ entries corresponding to leaf edges, half  of which correspond to the variables $a^j_C$ and half to the homogenizing variables $a^j_A$. Let the entries of $v_i$ corresponding to leaf edges $a^j_C$ be equal to $\alpha_i$, to leaf edges $a^j_A$ equal to $\beta_i$, and set all other entries equal to zero.

Then for $1 \leq j \leq r$, $1 \leq b \leq q$,
$$
l_{i,b}(v_i)  = \alpha_i t_b + (n-t_b)\beta_i = (\alpha_i - \beta_i)t_b + n\beta_i.
$$
Notice that this function depends only on the number of leaf edges in the subforest of $T_i$ induced by $b$. Let 
$ \mu_i : \rr \rightarrow \rr$ be given by 
$$
t \mapsto   (\alpha_i - \beta_i)t + n\beta_i.
$$ 
The two parameters  $\alpha_i$ and $\beta_i$ allow us to make $\mu_i$ whatever line we wish in $\rr^2$. Now we have $l_{i,b}(v_i) = \mu_i(t_b)$, and as explained, this value is completely independent of the trees under consideration. Thus, if $\mu_i(t) < \mu_j(t)$ for all $j \not = i$, then for any $b$ with $t_b =t$, $i$ wins $b$ at $v$. Choose $\alpha_i$ and $\beta_i$  so that  $ {\rm min}_j (\mu_j(t))$ is a continuous piecewise linear function,
$ {\rm min}_j (\mu_j(t)) = \mu_1(t)  \text{ if }t \in [0, \frac{5}{2}]$, and for $1< i \leq r$, $ {\rm min}_j (\mu_j(t))= \mu_i(t)  \text{ if } t \in [2i - \frac{3}{2}, 2i + \frac{1}{2}] $. Then 1 wins $b$ at $v$ if $t_b=0,2$, and for $1 < i \leq r$, $i$ wins $b$ at $v$ if $t_b = 2i$
 or  $2i -1$. 
 
 Form the matrices $M_i(v)$ with columns equal to the vectors in $D_i(v)$.  Now we just need to show that for $1 \leq i \leq r $, rank($M_i(v)) = (2n - 3) + 1 $. In order to do so, we will reinterpret our matrices in order to utilize previous results about reconstructing trees from subtree weights. 
 Let $T'$ be a tree, and assign to each edge a positive 
weight $w(e)$. Define the weight of a subforest to be the sum of the weights of the edges contained in the subforest. Let $M'_i(v)$ be the matrix consisting of the first $2n-3$ rows of $M_i(v)$ and

%$$
%M(t_b) =
%\begin{cases}
%\mu_1(t_b) & \text{if }t_b \in [0, 2.5] \\
%\mu_2(t_b) & \text{if }t_b \in [2.5, 4.5] \\
%\mu_3 (t_b) & \text{if }t_b \in [4.5, 6.5] \\
%\vdots \\
%\mu_r(t_b) & \text{if } t_b \in [2r - \frac{3}{2}, 2r + \frac{1}{2}] \\
%\end{cases}
%$$

%As an example of this construction, suppose, $r=3$ and $n=6$.
%Then Figure \ref{TropEven} shows a graph of $M$. In this case, we chose the slope of $\mu_i$ to be equal to $\dfrac{1}{2^{i-1}}$ and solved for $v$ accordingly.

%\begin{figure}
%\centering
%\begin{minipage}{.49\textwidth}
%  \centering
%  \includegraphics[width=.3\linewidth]{TropEven.pdf}
%  \caption{$M(t_b), n$ even}
%  \label{TropEven}
%\end{minipage}
%\begin{minipage}{.49\textwidth}
%  \centering
%  \includegraphics[width=.3\linewidth]{TropOdd.pdf}
%  \caption{$M(t_b), m$ odd}
%  \label{TropOdd}
%\end{minipage}
%\end{figure}

%\begin{align*}
%v_1 &= ( 1, 1, 1, 1, 1,1,0, 0,  0, 0, 0,  0,  0, 0,  0, 0, 0, 0) \\
%v_2 &= ( \frac{5}{24}, \frac{5}{24}, \frac{5}{24}, \frac{5}{24}, \frac{5}{24}, \frac{5}{24}, 0, 0,  0, \frac{17}{24}, \frac{17}{24},  \frac{17}{24},  \frac{17}{24}, \frac{17}{24}, \frac{17}{24}, 0, 0 ,  0) \\
%v_3 &= ( \frac{19}{48}, \frac{19}{48}, \frac{19}{48}, \frac{19}{48}, \frac{19}{48}, \frac{19}{48}, 0,  0,  0, \frac{31}{48}, \frac{31}{48},  \frac{31}{48},  \frac{31}{48}, \frac{31}{48}, \frac{31}{48} , 0 , 0 , 0). \\
%\end{align*} 

 $$
 w =
 \begin{pmatrix}
w(e_1) \\ 
\vdots \\ 
w(e_{2n-3}) \\
\end{pmatrix},
$$
then $ M'_i(v)w $ is a column vector with $j$-th entry equal to the weight of the subforest corresponding to the $j$-th column of 
$M'_i(v)$. 
 
$M'_1(v)$ contains column vectors corresponding to the empty 
subforest as well as the subforests with two leaf edges. A subforest with exactly two leaf edges with degree one vertices $u$ and $v$  is just the path between $u$ and $v$ with weight $d(u,v)$. Therefore, for fixed $w$, the entries of the column vector $M'_1(v)w$ determine a tree metric $\delta$ which has a graph realization $T'$.  By the \emph{Tree-Metric theorem} (\cite{Pachter, SempleandSteel}) $T'$ is the unique tree metric representation of $\delta$, and $w$ is the unique solution to 
$$
M'_1(v)x = M'_1(v)w.
$$
Therefore, we know rank($M'_1(v)) = 2n - 3 $. In fact, we have that the column rank of just the columns corresponding to subforests with two leaf edges is equal to $2n - 3$. If we let  $(x_1, \ldots, x_{2n-3}, y_1, \ldots, y_{2n-3})$ be an arbitrary vector in  $ \rr^{4n-6}$, then each of the columns is contained in the subspace defined by
$$
 x_1 + \ldots + x_{n} = \frac{2}{n-2}( y_1 + \ldots + y_{n}) .
 $$
The column corresponding to the empty subforest is clearly not contained in this subspace, so its addition increases the column rank by one, which implies  
rank($M_1(v)) = 2n - 2$. 

For $i \geq 2$, in order to show that rank($M_i(v)) = 2n-2$, we will first show that we can recover every edge weight if we know the weight of the subforests on $2i$ and $2i -1$ leaves. To determine the weight of a leaf edge $uv$ with degree three vertex $u$, choose a subforest with $2i$ leaves that includes all three edges incident to $u$. Choosing such a subforest is always possible since $2i \geq 3$. Removing $uv$ results in a subforest with $2i-1$ leaves with corresponding vector also in $D_i(v)$.The difference of the weights of these two subforests determines the weight of the leaf edge. 
 
For an internal edge $uv$, we construct a subforest that includes the edge $uv$ and the other four edges incident to either $u$ or $v$. The fact that $2i \geq 4$ ensures that such a subforest exists. Again, omitting $uv$ from this subforest gives us a different subforest on $2i$ leaves, and subtracting, we obtain the weight of $uv$.

For each edge, we found two subforests that differed by exactly that edge. Subtracting these vectors we obtain every column of the matrix 
$
\begin{pmatrix}
I \\
-I \\
\end{pmatrix},
$ where $I$ is the $(2n-3) \times (2n-3)$ identity matrix.
Anything in the column span of these vectors possesses the property that the entry for $a_C^k$ is just the negative of the entry for $a_A^k$. Therefore, adding any vector without this property to the set, i.e., adding any of the other subforest vectors, increases the rank by one. Thus, rank($M_i(v))=2n - 2.$

%$$\begin{pmatrix} 1 & 0 & 0 & 0 & 0 & 0 & 0 & 0 & 0 \\ 
%0 & 1 & 0 & 0 & 0 & 0 & 0 & 0 & 0 \\ 
%0 & 0 & 1 & 0 & 0 & 0 & 0 & 0 & 0 \\ 
%0 & 0 & 0 & 1 & 0 & 0 & 0 & 0 & 0 \\ 
%0 & 0 & 0 & 0 & 1 & 0 & 0 & 0 & 0 \\ 
%0 & 0 & 0 & 0 & 0 & 1 & 0 & 0 & 0 \\ 
 %0 & 0 & 0 & 0 & 0 & 0 & 1 & 0 & 0 \\ 
 %0 & 0 & 0 & 0 & 0 & 0 & 0 & 1 & 0\\ 
 %0 & 0 & 0 & 0 & 0 & 0 & 0 & 0 & 1  \\ 
%-1 & 0 & 0 & 0 & 0 & 0 & 0 & 0 & 0 \\ 
  %0 & -1 & 0 & 0 & 0 & 0 & 0 & 0 & 0 \\ 
 %0 & 0 & -1 & 0 & 0 & 0 & 0 & 0 & 0 \\ 
 %0 & 0 & 0 & -1 & 0 & 0 & 0 & 0 & 0 \\ 
 %0 & 0 & 0 & 0 & -1 & 0 & 0 & 0 & 0 \\ 
 %0 & 0 & 0 & 0 & 0 & -1 & 0 & 0 & 0 \\ 
 %0 & 0 & 0 & 0 & 0 & 0 & -1 & 0 & 0 \\ 
 % 0 & 0 & 0 & 0 & 0 & 0 & 0 & -1 & 0\\ 
 %0 & 0 & 0 & 0 & 0 & 0 & 0 & 0 & -1  \\ 
%\end{pmatrix}$$

\bigskip
\noindent \emph{Case 2:} $r$ is odd
\bigskip 

\noindent We construct the vector $v = (v_1, \ldots, v_r)$ as in the first case so that $ {\rm min}_j (\mu_j(t))$ is a continuous piecewise linear function with
$$
 {\rm min}_j (\mu_j(t)) =
\begin{cases}
\mu_1(t) & : t \in [0, \frac{5}{2}] \\
\mu_2(t) & : t \in [\frac{5}{2}, 4] \\
\mu_3(t) & : t \in [4, \frac{11}{2}] \\
\mu_i(t) &  : t \in [2i - \frac{5}{2}, 2i - \frac{1}{2}]. \\
\end{cases}
$$

 Notice that no $i$ wins $v$ at $b$ if $t_b=4$, but that if $t_b \not = 4$, then 
\begin{itemize}
\item 1 wins $b$ at $v$ if $t_b = 0$ or $2$
\item 2 wins $b$ at $v$ if $t_b = 3$
\item 3 wins $b$ at $v$ if $t_b = 5$
\item for $ 3< i \leq r$, $i$ wins $b$ at $ v$ if $t_b = 2i - 2$ or $t_b = 2i -1.$ 
\end{itemize}

Let $s^3_b$ be the number of internal edges of the subforest induced by $b$ on $T_3$.  We will perturb the entries of $v_3$ so that none of the above winning directions is affected but so that if $t_b =4$, then $3$ wins $b$ at $v$ if $s_b^3 = 0$ and $4$ wins $b$ at $v$ if $s_b^3 > 0$.

Set the entries of $v_3$  corresponding to internal edges $a^j_C$ equal to $\gamma$ and to internal edges $a^j_A$ equal to $\delta$.  Let 
$\tilde \mu_3 (s,t) =  \mu_3(t) + \gamma s + (n - 3 - s)\delta$. 
Now $l_{3,b}(v_3) = \tilde\mu_3(s^3_b, t_b)$ and for $i \not = 3$, $l_{i,b}(v_i) = \mu_i(t_b)$.
 Let $\delta = \epsilon > 0$ and
  $\gamma = - \epsilon(n - 3.9) $.
Now if $t_b = 4$ and $s_b = 0$, then 
   $l_{3,b}(v_3) = \tilde \mu_3(0,4) =  \mu_3(4) + \epsilon(n - 3) = \mu_2(4) +  \epsilon(n - 3) = l_{2,b}(v_2) + \epsilon(n - 3) > l_{2,b}(v_2)$, and so 2 wins $b$ at $v$. 
   
By our choice of $\gamma$ and $\delta$, as $s$ increases $\tilde \mu_3(s,t)$ decreases. Therefore, to show that 3 wins $b$ at $v$ when  $t_b = 4$ and $s_b > 0$, it is enough to show that 3 wins $b$ at $v$ when $t_b = 4 $ and $s_b = 1.$ In that case, we have
   $l_{3,b}(v_3) =\tilde \mu_3(1,4) = \mu_3(4) + \gamma + \epsilon(n - 4) < \mu_3(4) = \mu_2(4) = l_{2,b}(v_2)$, which implies 3 wins $b$ at $v$. Now choose $\epsilon > 0 $ small enough to ensure that the winning directions for  $t_b \not = 4$ are unaffected.

\noindent

Now we would like to show that for $1 \leq i \leq r$, ${\rm dim}_{\rr} \langle D_i(v) \rangle_\rr = 2n - 2.$ When $i = 1$, the proof is the same as in the even case. Likewise, for $i > 3$, $D_i(v)$ contains vectors corresponding to subforests with $2i - 2$ and $2i -1$ leaves. Essentially the same arguments from the even case where  $D_i(v)$ contained vectors corresponding to subforests with $2i - 1$ and $2i $  leaves show that ${\rm dim}_{\rr} \langle D_i(v) \rangle_\rr = 2n - 2.$ Therefore, we just need to establish the rank of $M_2(v)$ and $M_3(v)$.

$M'_2(v)$ contains vectors corresponding to every 3-leaf subtree of $T_2$. Just as before, given the weight vector $w$, $M'_2(v)w$ encodes the weight of every 3-leaf subtree. The main theorem in \cite{Pachter} implies that these weights uniquely determine $w$, so rank($M'_2(v)) = 2n - 3$. Every one of these vectors is contained in the hyperplane defined by 

$$
 x_1 + \ldots + x_{n} = \frac{3}{n - 3}( y_1 + \ldots + y_{n}).
 $$
Since any binary tree with four or more leaves contains at least two cherries, $T_3$ has a 4-leaf subforest with no internal edges, and $2$ wins $v$ at the coordinate corresponding to this subforest. The vector corresponding to this subforest is not contained in the hyperplane above, so
rank($M_2(v)) = 2n-2$. 

Just as before, for each edge $e$ of $T_3$ we use two subforests from the set of 4 and 5-leaf subforests that differ only by $e$ to determine $w(e)$. To carryout this procedure, we never require a subforest with no internal edges. This would only be the case if when isolating an internal edge $e$ with endpoints $u$ and $v$, every other edge adjacent to $u$ and $v$ was a leaf edge. However, that would imply that $n = 4$, which is contradiction. Again, the columns of the matrix 
$
\begin{pmatrix}
I \\
-I \\
\end{pmatrix}
$ are in the column span and adding any of the original subforest vectors gives us
rank($M_3(v)) = 2n -2$.
\end{proof}

%----------------------------------------------------------------------------------- 
%
%    PHYLOGENETIC INVARIANTS
%%----------------------------------------------------------------------------------- 

\section{Phylogenetic Invariants}
\label{Phylogenetic Invariants}

As outlined, the fact that $ D(3) = 6 $ means that to show the generic identifiability of the tree 
 parameters for the 3-tree Jukes-Cantor mixture model, we only need to separate all mixtures with $n=6$ leaves. While this is at least finite, there are still
 $19,698,048,370$ pairs of 6-leaf triplets, making it infeasible to try
 to list all pairs and separate them directly.
For some pairs $S, T \in \mathcal{T}_{ [6], 3}$, there is a 
disentangling set $K$ such that $|K| < 6 $, so our strategy is to first generate all pairs of 5-leaf triplets and wherever possible show the 
 mutual non-containment of their corresponding varieties. 
 By application of Theorem \ref{nondefective}, we actually only need to 
 show the two varieties are not identical, and to do so we will need to find a phylogenetic invariant that holds for one mixture and not the other. 
 Once complete, we will have a short list of 5-leaf triplet pairs 
 for which we are unable to show that their varieties are not identical. 
 From this list, we arrive at a much smaller list of pairs of 
 $6$-leaf triplets which we need to separate using invariants.
 For these $6$-leaf triplets, we use linear invariants and
 higher order invariants to separate all of the pairs.

This final step in the proof of Theorem \ref{MainThm} is
highly computational.  The steps are all completely contained
in the three Maple (\cite{Maple}) worksheets  
\begin{center}
{\tt LinearInvariants\_5Leaf.mw}, {\tt LinearInvariants\_6Leaf.mw}, 

{\tt Higher\_Degree\_Invariants.mw}
\end{center}  
which are located at the website:
\begin{center}
{\tt http://www4.ncsu.edu/~smsulli2/Pubs/ThreeTreesWebsite/threetrees.html}
\end{center}

We outline the methods for finding separating
linear invariants and higher-order invariants in the next two subsections.

\subsection{Linear Invariants}
\label{Linear Invariant}

Our first step will be to distinguish 3-tree mixture models with 
5 leaves from one another by finding linear invariants that hold for one mixture but not the other. A few
observations will help us reduce the dimension of the ambient space of the varieties. For our purposes, it is unnecessary to calculate invariants that hold for all mixtures. For the Jukes-Cantor model, because $\{ C, G, T \}$ are in the 
same class, every model will have linear invariants arising from 
permutations on this set. For example, any five leaf mixture 
will have the same parameterization on the coordinates in the set 
$$
\{ q_{CCCGT}, q_{GGGCT}, q_{TTTGC}, q_{CCCTG}, 
q_{GGGTC}, q_{TTTCG} \} .
 $$
 Therefore, the difference of any two of these elements is an invariant for every 3-tree mixture model with 5 leaves.
We will only consider the lexicographically first element as a 
representative of each such set. By doing this, and removing
coordinates that are always zero, we can perform our
calculations in $ \cc^{51} $ instead of $\cc ^ {1024} $. (In the provided maple file, we also exclude the coordinate $q_{AAAAA},$ which is not involved in any linear invariant).

Applying an element of $S_5$ to the leaf labels of all of the 
trees in a mixture model merely permutes the coordinates of 
our parameterization. As a result, once we have determined that the models corresponding to a 5-leaf triplet pair do not in fact 
have the same variety, then by applying the elements of $S_5$
to all six trees in the pair, we generate new 5-leaf triplet pairs
that do not have the same variety. To illustrate, although there
are 680 different identical 5-leaf triplet pairs ($(T,T) \in \mathcal{T}_{[n],r}\times \mathcal{T}_{[n],r}$), all of 
these pairs can be
generated by applying an element of $S_5$ to one of only 28
such pairs. 

The following lemma allows us to compute the linear invariant space of every 3-tree mixture model using the built-in Maple function {\tt{IntersectionBasis}}.

\begin{lemma} 
A linear polynomial $ f $ is an invariant of $V_{T_1} * V_{T_2} * V_{T_3} $ $ \iff $ $ f $ is an invariant of $V_{T_1}, V_{T_2},$ and $ V_{T_3} $.
\end{lemma}
 
In the provided code, we compare all 5-leaf triplet pairs, and 
up to the action of $S_5$, there are 36 pairs with the same linear invariant space. As mentioned, this list consists of the 28 pairs of identical triplets as well as 8 additional pairs containing two distinct triplets.

If there exists $S,T \in  \mathcal{T}_{ [6] , 3}$ such that $V_{T_1}*V_{T_2}*V_{T_3} = V_{S_1}*V_{S_2}*V_{S_3}$, then it must be 
the case that for any five element subset $K$, $(S_{|K} ,T_{|K})$ (or some permutation of the leaves thereof) is one of the 5-leaf 
triplet pairs in our list. Therefore, the only 6-leaf triplet pairs that are candidates for inseparability are those generated by 
attaching an additional edge to each of the 6 trees in an inseparable 5-leaf triplet pair. Since there are 36 pairs, and each 5-leaf tree has 7 edges, the number of 6-leaf triplets we must consider is less
 than  $(36)(7^6) = 4,235,364$. This is far fewer than would be expected had we proceeded directly to the 6-leaf case. 

Just as in the 5-leaf case, nondefectiveness of all of the involved varieties is ensured by Theorem \ref{nondefective}, so it suffices to show that  $V_{T_1}*V_{T_2}*V_{T_3} \not = V_{S_1}*V_{S_2}*V_{S_3}$. This fact is 
reflected in our computation of the set {\tt{AllSixLeafPairs}}, which contains the eighty-five 6-leaf triplet pairs (up to the action of 
$S_6$) that are not separated by any linear invariant.
% $I(V_{T_1}*V_{T_2}*V_{T_3})_1 = I ( V_{S_1}*V_{S_2}*V_{S_3})_1 .$
For these eighty-five pairs of $6$-leaf triplets we must
find higher degree invariants that separate them.

\subsection {Invariants of higher degree}
\label{Invariants of Higher Degree}

If we wish to determine identifiability, we must broaden our search to invariants of higher degree. Unfortunately, we glean little from studying $V_{T_1}, V_{T_2}$, and $V_{T_3}$ individually, 
as there is no reason to expect that a non-linear invariant that holds for all three holds for their join.

\begin{figure}
\includegraphics[width=5.2cm]{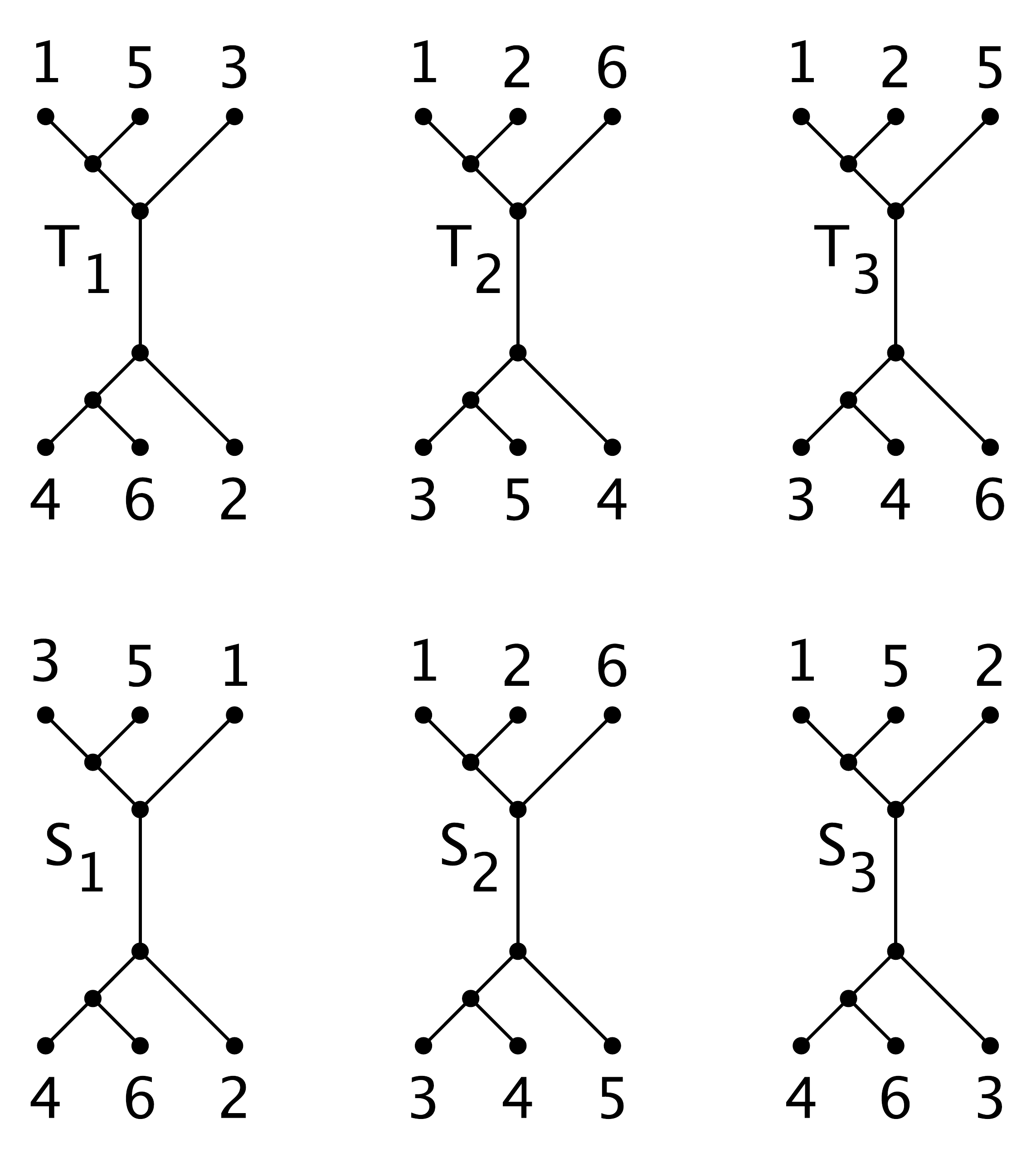}
\caption{A 6-leaf triplet pair that is not separated by linear invariants.}
\label{6LeafExample_Both}
\end{figure}

After removing trivial coordinates and linear invariants that hold for all 6-leaf trees we can perform our calculations for 6-leaf trees in 186-dimensional image space. When searching for linear invariants, we will also disregard any coordinate in which the character A appears in the index so that the bulk of our computations are done in a 31-dimensional image space.  If we let $T$ be the six leaf 
triplet described in Figure \ref{6LeafExample_Both} (in the worksheet, these trees are labelled $T_{16}, T_{19}$ and $T_{63}$), there are 61 linearly independent elements of 
$I(V_{T_1}*V_{T_2}*V_{T_3})_1$ which define a 125-dimensional linear subspace containing $V_{T_1}*V_{T_2}*V_{T_3}$. Since dim($V_{T_1}*V_{T_2}*V_{T_3}) = 29$, at least for this particular triplet, higher-degree invariants must exist. In order to find these, we will use the explicit parameterization of the variety in Fourier 
coordinates. Here, for example, is the parameterization of $T$ on just the first three coordinates. As mentioned in Section \ref{Fourier}, we make the expressions more compact by letting $a_{C}^{i|[6]\setminus i} = a_{i}$, $b_{C}^{i|[6]\setminus i} = b_{i}$, $c_{C}^{i|[6]\setminus i} = c_{i}$,
  $a_C^{15|2346} = a_{7}$, $a_C^{15|2346} = a_{8}$, $a_C^{135|246} = a_{9}$, $b_C^{12|3456} = b_{7}$, $b_C^{15|2346} = b_{8}$, $b_C^{126|345} = b_{9}$, $c_C^{12|3456} = c_{7}$, $c_C^{34|1256} = c_{8}$, and $c_C^{125|346} = c_{9}$.
 \begin{align*}
 q_{AAAACC} &= \pi_1  a_{5} a_{6} a_{7} a_{8} a_{9}  + 
  \pi_2  b_{5} b_{6} b_{8} b_{9}   + 
  \pi_3  c_{5} c_{6} c_{9}  \\
 q_{AAACAC} &=\pi_1 a_4 a_6  +  \pi_2 b_4 b_6 b_9  + \pi_3 c_4 c_6 c_8  \\
q_{AAACCA} &= \pi_1 a_4 a_5 a_7 a_8 a_9  +  \pi_2 b_4 b_5 b_8  +  \pi_3 c_4 c_5 c_8 c_9 \\
 \end{align*} 
\noindent Theoretically, we should be able use these equations to construct an ideal in 
$$\cc[y_1, \ldots, y_{186}, a_1, \ldots, a_9, b_1, \ldots, b_9, c_1, \ldots, c_9, \pi_1, \pi_2],$$ 
and eliminate to obtain a Gr\"obner basis for $I(V_{T_1}*V_{T_2}*V_{T_3})$ that includes non-linear invariants.
The number of variables and equations involved in the computation make it
infeasible.

Instead, we will apply heuristic methods to reduce the number of 
variables in the ideal before we attempt elimination. Our strategy 
will be to find linear invariants that hold for $V_{T_2}$ and $V_{T_3}$, but not for $V_{T_1}$ . The resulting equations will not evaluate to zero on all of $V_{T_1}*V_{T_2}*V_{T_3}$, but will not involve the parameters from $T_2$ or $T_3$ at all. To illustrate, we revisit the trees in Figure \ref{6LeafExample_Both}, and look at a few more coordinates.
\begin{align*}
q_{CCCAAC} &= \pi_1 a_1 a_2 a_3 a_6 a_7 a_8 + \pi_2 b_1 b_2 b_3 b_6 b_8 b_9 + \pi_3 c_1 c_2 c_3 c_6 c_8 \\
q_{CCGAAG} &= \pi_1 a_1 a_2 a_3 a_6 a_7 a_8 a_9 + \pi_2 b_1 b_2 b_3 b_6 b_8 b_9 + \pi_3 c_1 c_2 c_3 c_6 c_8 \\
q_{CACCGT} &= \pi_1 a_1 a_3 a_4 a_5 a_6 a_7 a_8 a_9 + \pi_2 b_1 b_3 b_4 b_5 b_6 b_7 b_8 b_9 + \pi_3 c_1 c_3 c_3 c_4 c_5 c_6 c_7 c_9 \\
q_{CAGGTG} &= \pi_1 a_1 a_3 a_4 a_5 a_6 a_7  + \pi_2 b_1 b_3 b_4 b_5 b_6 b_7 b_8 b_9 + \pi_3 c_1 c_3 c_3 c_4 c_5 c_6 c_7 c_9, \\
\end{align*}
Now it is easy to spot linear invariants for $V_{T_2}$ and $V_{T_3}$, and subtracting we obtain
\begin{align*}
q_{CCCAAC} - q_{CCGAAG} &= \pi_1 a_1 a_2 a_3 a_6 a_7 a_8 - \pi_1 a_1 a_2 a_3 a_6 a_7 a_8 a_9 \\
q_{CACCGT} - q_{CAGGTG} &= \pi_1 a_1 a_3 a_4 a_5 a_6 a_7 a_8 a_9 - \pi_1 a_1 a_3 a_4 a_5 a_6 a_7.  
\end{align*}

In this particular case, $\mathrm{dim}(I(V_{T_1}*V_{T_2}*V_{T_3})_1) - \mathrm{dim}(I(V_{T_2}*V_{T_3})_1) = 20$, so there are twenty linearly 
independent relations only involving the parameters from $T_1$. 
We introduce new variables for the image space and use these relations to construct the ideal,
$$
 J = \langle y_{1} -( \pi_1 a_1 a_2 a_3 a_6 a_7 a_8 - \pi_1 a_1 a_2 a_3 a_6 a_7 a_8 a_9), y_2 - (\pi_1 a_1 a_3 a_4 a_5 a_6 a_7 a_8 a_9 - \pi_1 a_1 a_3 a_4 a_5 a_6 a_7), \ldots \rangle,
 $$
 where $ J \subseteq \cc[y_1, \ldots, y_{20}, a_1, \ldots, a_9, \pi_1]$.
With fewer parameters, we can now compute a Gr\"obner basis for $J \cap \cc[y_1, \ldots, y_{20}]$ using elimination in Macaulay2 \cite{M2}. The basis gives us a set of relations in the $y_i$ variables, which we translate back into our original coordinates. For this particular triplet, we find
$$
(q_{CCCAAC} - q_{CCGAAG})(q_{CACCGT} - q_{CAGGTG}) =
(q_{CCCAGT} - q_{CCGATC}) (q_{CACCAC} - q_{CAGGAC}).
$$ 

\noindent Finally, to separate the triplet pair from Figure \ref{6LeafExample_Both}, we substitute the parameterization of $S$ into this relation and confirm that it does not evaluate to zero.

This technique allows us to find an invariant for one mixture that does not hold for the other for all of the triplet pairs contained in {\tt{AllSixLeafPairs}}. As outlined, the existence of these invariants is sufficient to establish the generic identifiability of the tree parameters of the 3-class Jukes-Cantor mixture model.
Among the supplementary materials is the worksheet {\tt{HigherDegreeInvariants}} which lists an invariant separating each triplet pair and provides code to quickly  generate the coordinate functions for verification.

\section*{Acknowledgments}

Colby Long was partially supported by the US National Science Foundation (DMS 0954865).
Seth Sullivant was partially supported by the David and Lucille Packard 
Foundation and the US National Science Foundation (DMS 0954865).

%----------------------------------------------------------------------------------- 
%
%    BIBLIOGRAPHY
%%----------------------------------------------------------------------------------- 
 \bibliography{references}
\bibliographystyle{plain}

\end{document}